\documentclass[11pt,a4paper]{article}
\usepackage{amsmath,amssymb,amsthm}
\usepackage{booktabs}
\usepackage[margin=2.6cm]{geometry}
\usepackage[T1]{fontenc}
\usepackage{tgtermes}
\usepackage{tgheros}
\usepackage[british]{babel}
\usepackage{microtype}
\usepackage[hidelinks]{hyperref}
\hypersetup{pdftitle={Discrete energy as an exact label-free training objective for finite-element surrogates}, pdfauthor={Ruifeng Cao, Xidan Song}, pdfkeywords={finite element method; surrogate modelling; label-free training; minimum total potential energy; energy-norm conditioning; JEPA; executable falsification}}

\newtheorem{lemma}{Lemma}
\newtheorem{proposition}{Proposition}
\newtheorem{corollary}{Corollary}
\theoremstyle{remark}
\newtheorem{remark}{Remark}

\newcommand{\norm}[1]{\left\lVert #1 \right\rVert}
\newcommand{\Kff}{K_{\mathrm{ff}}}
\newcommand{\gap}{\operatorname{gap}}

\title{Discrete energy as an exact label-free training objective\\
for finite-element surrogates}
\author{Ruifeng Cao\\[2pt]
{\small The University of Manchester, Manchester, United Kingdom}\\
{\small \href{mailto:ruifeng.cao@manchester.ac.uk}{\texttt{ruifeng.cao@manchester.ac.uk}}}\\
{\small ORCID: \href{https://orcid.org/0000-0001-9077-9561}{0000-0001-9077-9561}}
\and
Xidan Song\\[2pt]
{\small Wuhan University, Wuhan, China}\\
{\small \href{mailto:liuqingdamafen@whu.edu.cn}{\texttt{liuqingdamafen@whu.edu.cn}}}\\
{\small ORCID: \href{https://orcid.org/0000-0003-2612-6296}{0000-0003-2612-6296}}}
\date{5 August 2026}

\begin{document}
\maketitle

\begin{abstract}
Supervised training of finite-element (FE) surrogate models requires
reference solutions, and each reference solution is obtained by solving the
system that the surrogate is intended to replace. The assembled discrete
potential energy provides a training signal that requires no reference
solution. This note records, with proofs, the identities that make this
signal exact for linear elastostatics: the difference between the energy of
a prediction and the energy of the reference solution equals one half of
the squared stiffness-norm error, and the gradient of the energy equals the
stiffness-weighted error. Label-free discrete-energy minimisation and
supervised regression in the stiffness norm therefore have the same unique
minimiser and identical gradients at every point. Around this central
result, the note states a conditioning lemma that bounds the displacement
error by the energy gap, a modewise contraction identity that explains why
the Euclidean displacement error is an unsuitable primary metric, the
Chebyshev bound that governs conjugate-gradient post-processing of
surrogate predictions, and a conditional latent-separation proposition for
joint-embedding predictive architecture (JEPA) pretraining on a shared
stiffness operator, with an explicit numerical counterexample that delimits
its scope. Every claim with numeric content is implemented as an executable
falsification check; the checks were executed twice, on synthetic test
problems and on a probe set of 16 instances from the validation split of a
pre-registered experimental run, and every inequality holds, with the
measured tightness reported. A closing section explains why the
construction does not extend to elastodynamics through direct minimisation
of the action functional, and which time-discrete formulation restores
exactness.

\end{abstract}

\section{Introduction}

Neural surrogate models for finite-element (FE) analysis are intended to
reduce the cost of repeated solves. After training, a single network
evaluation replaces an assembly-and-solve cycle, so that design iteration,
optimisation studies, and uncertainty quantification can obtain approximate
mechanical responses at a small fraction of the original cost. The standard
training procedure, however, is supervised regression on reference
solutions, and each reference solution is computed by one full solve of
precisely the system that the surrogate is intended to replace. This
creates a circular cost structure. A training corpus of between roughly
$10^{3}$ and $10^{5}$ instances must be generated by the original solver
before the surrogate produces its first prediction; the computational cost
of this corpus must be recovered by later use of the surrogate before any
net saving is achieved; and every extension of scope, such as a new
geometry family, a new load regime, or a finer mesh, requires additional labelled data, that is, additional instances paired with reference solutions, and therefore additional solves. In practice, the supply of
labels, rather than the network architecture, is the principal constraint
on FE surrogate modelling.

For linear elastostatics this constraint is not fundamental, because the
discretised problem carries a training signal whose evaluation requires no
reference solution: the assembled discrete potential energy
$\Pi_h(u) = \tfrac12\,u^{\!\top}\Kff\,u - F^{\!\top}u$, whose setting is
made precise in Section~\ref{sec:setting}. The two ingredients of this
expression, the stiffness operator $\Kff$ and the load vector $F$, are
assembled from the mesh, the material parameters, and the load data; this
assembly is performed by every FE preprocessor before a solver is invoked.
Evaluating $\Pi_h$ at a surrogate prediction costs one sparse
matrix--vector product and two inner products. No reference solution
appears anywhere in the expression. The value $\Pi_h(U^\ast)$ at the
reference solution $U^\ast$ is unknown, but it is a constant of the problem
instance and is never required, because adding a constant to an objective
changes neither its minimisers nor its gradients. A surrogate can therefore
be trained on arbitrarily many unlabelled instances at assembly cost.
Two questions must then be answered before this signal can be trusted. The
first is what, exactly, minimising $\Pi_h$ trains the surrogate towards,
and in which norm the resulting error is controlled. The second is what
consequences the answer has for how such surrogates should be evaluated,
post-processed, used for representation learning, and extended to
dynamics. This note answers both questions on the reference discretisation itself. Throughout, exactness refers to the training objective relative to the chosen finite-element discretisation; it is not a claim about optimisation convergence, network expressivity, or generalisation to unseen instances.

The reason that $\Pi_h$ can substitute for labelled data is the
principle of minimum total potential energy, referred to below as the
minimum principle. On the discrete space, the reference solution is the unique
minimiser of $\Pi_h$. The functional therefore characterises the solution
implicitly, and decreasing the value of $\Pi_h$ decreases the distance to
the solution in a norm that Lemma~\ref{lem:one} identifies exactly.
Wherever the governing physics is described by a minimum principle,
minimising a variational energy is consequently a natural route to
training without supervision, and this route has been developed in several
directions. The Deep Ritz method \cite{EYu2018} minimises the continuous
energy functional over a neural-network trial class. The Deep Energy
Method \cite{Samaniego2020} applied the same idea to problems in
computational mechanics. The associated numerical analysis
\cite{LiaoMing2021,LiuCaiRamani2023} identifies the sources of error in
this approach: the total error decomposes into a network approximation
term together with quadrature and generalisation terms, because the
continuous functional must itself be approximated by numerical
integration. For operators trained without labels across many problem
instances, variational operator learning (VOL) \cite{XuVOL2024} operates
directly on the assembled finite-element system in a matrix-free manner
and uses the gradient identity as a computational device; its loss,
however, is the Euclidean norm of the residual, or regression to
provisional labels generated by a small number of conjugate-gradient
iterations, so that it does not minimise $\Pi_h$ itself, and it states no
equivalence or conditioning result. The variational physics-informed
neural operator (VINO) \cite{EshaghiVINO2025} evaluates the continuous
functional analytically on its own bilinear discretisation over uniform
grids, on which all element matrices are identical and no assembly is
required, treats voids by an immersed mask, and includes a mesh-refinement
study. In a separate line of work, the training difficulty of
residual-based physics losses has been attributed to operator
conditioning: the convergence rate of gradient descent is governed by an
operator associated with the Hermitian square, that is, the composition of the operator with its adjoint, of the underlying differential operator \cite{DeRyck2024}, whose condition number is the
square of that of the original operator. Underlying all of these lines is
the classical energy-norm algebra of symmetric positive definite (SPD)
systems \cite{Ciarlet1978,Saad2003}. What has been missing is the
statement, and the systematic exploitation, of its exact consequences for
surrogate training on the reference discretisation itself.

This note supplies that statement. The results are organised according to
how such a surrogate is built, evaluated, deployed, and extended, and each
technical ingredient is introduced because a specific question requires
it.

\textbf{Training.} Taking the assembled pair $(\Kff, F)$ of the reference
model as the training objective removes the cost of labels, but this step
is sound only if the resulting signal is exact rather than heuristic.
Lemma~\ref{lem:one} establishes exactness: minimising $\Pi_h$ is
supervised regression towards the solution of the reference solver, in the
norm induced by the stiffness operator, with identical gradients at every
point. Remark~\ref{rem:c2} adds an equivalent reading: the energy gap
equals the squared residual measured in the norm induced by $\Kff^{-1}$,
and the inverse cancels algebraically, which is why the loss can be
evaluated by assembly alone, without a solve and without a preconditioner.
Energy losses based on numerical quadrature of the continuous functional
cannot offer this exactness, because the quadrature error enters the
objective by construction.

\textbf{Evaluation.} Once training is label-free, evaluation cannot simply
default to the displacement error. On the probe instances of Section~\ref{sec:pass}, the condition number $\kappa$ of the stiffness
operator ranges from approximately $4.8\times10^{3}$ to
$1.5\times10^{5}$. At such condition numbers, a surrogate can be almost
consistent in energy and nevertheless exhibit a large relative Euclidean
($\ell_2$) displacement error; conversely, training supervised only on
displacement can leave large errors in the energy and therefore in the stresses, as measured in the energy norm. Lemma~\ref{lem:cond} and Corollary~\ref{cor:modes} turn this
observation into precise statements: the energy gap is the primary metric;
the displacement error is bounded by the energy gap multiplied by a factor
involving $\kappa$; and gradient descent on the energy contracts the error
in each stiffness eigenmode by exactly the factor $(1-\eta\lambda_i)$,
where $\eta$ denotes the step size and $\lambda_i$ the corresponding
eigenvalue. The same mechanism, namely that descent on a squared residual
squares the conditioning, explains anomalies reported for residual-trained
operators \cite{XuVOL2024}: smaller Euclidean residuals need not produce
smaller $L^2$ (mean-square) test errors, and direct residual minimisation performs
worst.

\textbf{Deployment.} A surrogate prediction can serve as the initial
iterate of an iterative solver. Appending $k$ conjugate-gradient (CG)
iterations on the assembled system converts the prediction into a solution
of controllable accuracy at known cost. The classical Chebyshev bound of
Section~\ref{sec:polish} bounds the effect of this post-processing step,
which this note refers to as polishing, before it is performed; the
measured margin between the bound and the observed contraction quantifies
how favourable the actual spectra are.

\textbf{Representation learning.} The same energy objective can supervise representation learning across load cases in joint-embedding predictive architecture (JEPA) pretraining \cite{LeCun2022,Assran2023}, a self-supervised paradigm in which an encoder maps each input to a latent code, that is, a low-dimensional internal representation, and a decoder maps latent codes back to predictions. A
known failure mode of such pretraining is representational collapse, that
is, the mapping of distinct inputs to nearly identical latent codes.
Whether the energy objective prevents collapse is the question of whether
accuracy provably forces latent codes apart. Proposition~\ref{prop:one}
answers this in the affirmative for load cases on a shared stiffness
operator, and an explicit, measured counterexample identifies the boundary
of validity: the analogous statement across different geometries fails, so
that separation claims across geometries require additional conditioning
information.

\textbf{Verification.} Every statement above carries numeric content, and
each is implemented as an executable check with thresholds fixed in
advance. The checks were executed twice: first on synthetic test problems
in the verification suite, and then on a probe set of 16 instances drawn from the 256-instance validation split of an experimental run whose full configuration was fixed in
advance and identified by a cryptographic hash. Under the pre-registered
protocol, a violated inequality would have required the withdrawal of the
corresponding claim. Table~\ref{tab:pass} reports the measured quantities.

\textbf{Extension.} None of the preceding results transfers directly to
elastodynamics, because Hamilton's principle is a stationarity principle
rather than a minimum principle. Section~\ref{sec:outlook} quantifies this
obstruction using the classical theory of conjugate points
\cite{GelfandFomin1963}: the time horizon on which the action functional
is convex along the discrete trajectory shrinks proportionally to the mesh
size, so that direct action minimisation becomes unsound precisely in the
refinement limit. The same section identifies the exact repair, namely an
incremental-variational time discretisation in which every implicit step
is again a convex minimisation and Lemma~\ref{lem:one} applies verbatim
step by step.

\section{Discrete elastostatic setting}\label{sec:setting}

Let a linear elastostatic problem be discretised by finite elements, and
let the degrees of freedom constrained by Dirichlet boundary conditions be
eliminated. The number of remaining free degrees of freedom is denoted by
$n$, and the reduced system reads $\Kff U^\ast = F$, where
$\Kff \in \mathbb{R}^{n\times n}$ is the stiffness operator restricted to
the free degrees of freedom (the subscript $\mathrm{ff}$ indicates the
free--free block of the full stiffness matrix), $F \in \mathbb{R}^{n}$ is
the corresponding load vector, and $U^\ast = \Kff^{-1}F$ is the reference
solution. The operator $\Kff$ is symmetric positive definite (SPD). The
assembled discrete potential energy is
\begin{equation}
\Pi_h(u) \;=\; \tfrac12\, u^{\!\top}\Kff\, u \;-\; F^{\!\top} u ,
\qquad u \in \mathbb{R}^{n}.
\label{eq:energy}
\end{equation}
For a surrogate prediction $\hat u \in \mathbb{R}^{n}$, the \emph{energy
gap} is defined as $\gap(\hat u) = \Pi_h(\hat u) - \Pi_h(U^\ast) \ge 0$,
and the stiffness norm (also called the energy norm) is defined by
$\norm{v}_K^2 = v^{\!\top}\Kff\, v$ for $v \in \mathbb{R}^{n}$. Since
$\Pi_h(U^\ast)$ is a constant of the problem instance, minimising
$\Pi_h(\hat u)$ requires no reference solve: the functional $\Pi_h$ is assembled from mesh, material, and load data alone. Note the asymmetry: the objective $\Pi_h(\hat u)$ is computable without a reference solve, whereas the energy gap differs from it by the unknown constant $\Pi_h(U^\ast)$ and is therefore an analysis and evaluation quantity; computing it, as in Section~\ref{sec:pass}, requires the reference solution. Throughout,
$\lambda_{\min}$ and $\lambda_{\max}$ denote the smallest and largest
eigenvalues of $\Kff$, and $\kappa = \lambda_{\max}/\lambda_{\min}$
denotes its spectral condition number. Sufficient essential boundary conditions are assumed throughout: if rigid-body modes are not eliminated, $\Kff$ is singular and the identities below must be restated on the corresponding quotient space.

\section{Exact equivalence of label-free and supervised energy-norm training}

\begin{lemma}[Label-free exactness]\label{lem:one}
For all $u \in \mathbb{R}^n$,
\begin{equation}
\Pi_h(u) - \Pi_h(U^\ast) \;=\; \tfrac12 \norm{u - U^\ast}_K^{2},
\qquad
\nabla \Pi_h(u) \;=\; \Kff\,(u - U^\ast).
\label{eq:identities}
\end{equation}
Consequently, the label-free objective $\Pi_h$ and the supervised objective $\tfrac12\norm{u-U^\ast}_K^2$ pose the same minimisation problem: both have the unique minimiser $U^\ast$, their values differ by
the constant $\Pi_h(U^\ast)$, and their gradients coincide at every point.
\end{lemma}

\begin{proof}
Expand $\tfrac12\norm{u-U^\ast}_K^2
= \tfrac12 u^{\!\top}\Kff u - u^{\!\top}\Kff U^\ast
+ \tfrac12 (U^\ast)^{\!\top}\Kff U^\ast$ and substitute
$\Kff U^\ast = F$. The first identity follows by comparison with
\eqref{eq:energy}, and differentiation gives the second.
\end{proof}

\begin{remark}[Verification]
Both identities in \eqref{eq:identities} are checked in the assembly-level
verification suite on assembled instances and hold to relative error
$\lesssim 10^{-11}$ in double precision.
\end{remark}

\begin{remark}[The energy gap equals the dual-norm residual]\label{rem:c2}
Let $r(u) = F - \Kff u$ denote the residual, and define the
$\Kff^{-1}$-norm by $\norm{r}_{K^{-1}}^2 = r^{\!\top}\Kff^{-1} r$. By
\eqref{eq:identities},
\[
\gap(u) \;=\; \tfrac12\,\norm{u - U^\ast}_K^2
\;=\; \tfrac12\, r(u)^{\!\top} \Kff^{-1}\, r(u)
\;=\; \tfrac12 \norm{r(u)}_{K^{-1}}^{2}.
\]
The energy gap therefore equals the residual measured in the
$\Kff^{-1}$-norm, which is the operator norm in which smallness of the
residual is equivalent to accuracy of the solution. The inverse cancels
algebraically, since $\Kff \Kff^{-1} \Kff = \Kff$, and this cancellation
is the reason why the loss is evaluable by assembly alone, without a solve
and without a preconditioner. The same observation identifies the boundary
of the construction: the cancellation is a property of the static,
self-adjoint, coercive operator. Section~\ref{sec:outlook} discusses what
changes for time-dependent operators.
\end{remark}

\section{Conditioning, modewise contraction, and the choice of metric}

\begin{lemma}[Conditioning bound]\label{lem:cond}
For all $u \in \mathbb{R}^{n}$,
\begin{equation}
\norm{u - U^\ast}_2 \;\le\; \sqrt{\frac{2\,\gap(u)}{\lambda_{\min}(\Kff)}},
\qquad\text{equivalently}\qquad
\frac{\norm{u-U^\ast}_2^2}{\norm{U^\ast}_2^2}
\;\le\; \kappa \cdot
\frac{\norm{u-U^\ast}_K^2}{\norm{U^\ast}_K^2}.
\label{eq:cond}
\end{equation}
\end{lemma}

\begin{proof}
The inequality $2\gap(u) = \norm{u-U^\ast}_K^2 \ge
\lambda_{\min}\norm{u-U^\ast}_2^2$ gives the first form. Combining it with
$\norm{U^\ast}_K^2 \le \lambda_{\max}\norm{U^\ast}_2^2$ gives the second.
\end{proof}

\begin{corollary}[Modewise contraction]\label{cor:modes}
Let $\Kff = \sum_i \lambda_i\, v_i v_i^{\!\top}$ be the eigendecomposition
of the stiffness operator, and let $u^{+} = u - \eta\,\nabla\Pi_h(u)$ be
one step of gradient descent with step size $\eta > 0$. Then the error
component $e_i = v_i^{\!\top}(u - U^\ast)$ in the $i$-th eigenmode evolves
as
\begin{equation}
e_i^{+} \;=\; (1 - \eta\,\lambda_i)\, e_i
\qquad\text{exactly, for every mode } i.
\label{eq:modes}
\end{equation}
\end{corollary}

\begin{proof}
Immediate from $\nabla\Pi_h(u) = \Kff(u-U^\ast)$ and diagonalisation.
\end{proof}

\begin{remark}[Choice of the primary metric]
By \eqref{eq:modes}, energy descent converges most slowly on the modes
with small eigenvalues, which are the smooth, low-frequency modes, and
these modes dominate the Euclidean displacement norm. The energy norm
weights the error in mode $i$ by $\lambda_i$ and is therefore the norm in
which progress is uniform across modes. As a consequence, a surrogate can
be nearly consistent in energy while exhibiting a substantial relative
$\ell_2$ displacement error, and training supervised only on displacement
can produce predictions whose energy gap, and therefore whose energy-norm stress error, is large. This explains the empirical observation recorded in Section~\ref{sec:pass}: at comparable displacement error, supervision on 1024 labelled instances attains a relative energy gap 5.1 times larger than label-free energy training.
\end{remark}

\section{Certified post-processing by conjugate gradients}
\label{sec:polish}

A surrogate prediction can be used as the initial iterate of the conjugate
gradient (CG) method applied to the assembled system. Appending $k$ CG
iterations is a post-processing step, referred to in this note as
polishing, which upgrades the prediction to a solution of controllable
accuracy at the cost of $k$ sparse matrix--vector products. Its effect on
the energy gap obeys the classical bound \cite{Saad2003}
\begin{equation}
\frac{\gap_k}{\gap_0} \;\le\; 4\,\rho^{2k},
\qquad
\rho = \frac{\sqrt{\kappa}-1}{\sqrt{\kappa}+1},
\label{eq:cheb}
\end{equation}
where $\gap_0$ denotes the energy gap of the prediction, $\gap_k$ the
energy gap after $k$ CG iterations initialised at the prediction, and
$\rho$ the Chebyshev contraction factor determined by the condition number
$\kappa$. Since $\gap = \tfrac12\norm{\cdot}_K^2$, the bound
\eqref{eq:cheb} is the square of the standard $K$-norm estimate for CG.
The bound is checked as an inequality on measured ratios
$\gap_k/\gap_0$. The margin between the measured ratios and the bound,
reported in Table~\ref{tab:pass}, shows that on the spectra of the
validation instances the observed contraction is considerably faster than
the worst case.

\section{JEPA representation learning: latent separation and its scope}\label{sec:prop1}

The role of the energy objective in preventing representational collapse in JEPA pretraining rests on the following geometric fact: distinct loads on a shared operator
have well-separated solutions, and any decoder that is simultaneously
accurate on all of them must therefore assign them distinct latent codes.
The next proposition makes this quantitative.

\begin{proposition}[Separation on a shared operator]\label{prop:one}
Fix one geometry, and hence one stiffness operator $\Kff$, and consider
loads $F_1,\dots,F_L$ with solutions $U_1^\ast,\dots,U_L^\ast$. Let
$z_1,\dots,z_L$ be latent codes and let the decoder $D$, that is, the map from the latent space to $\mathbb{R}^{n}$, be $L_D$-Lipschitz with respect to $\norm{\cdot}_K$, in the sense that
$\norm{D(z) - D(z')}_K \le L_D \norm{z - z'}$ for all latent codes
$z, z'$. If $\norm{D(z_i) - U_i^\ast}_K \le \varepsilon$ for all $i$,
which is equivalent to $\gap(D(z_i)) \le \varepsilon^2/2$, then for all
$i \ne j$,
\begin{equation}
\norm{z_i - z_j\,} \;\ge\;
\frac{\norm{U_i^\ast - U_j^\ast}_K - 2\varepsilon}{L_D}.
\label{eq:prop1}
\end{equation}
\end{proposition}

\begin{proof}
By the triangle inequality,
\[
\norm{U_i^\ast-U_j^\ast}_K \le \norm{U_i^\ast - D(z_i)}_K
+ \norm{D(z_i) - D(z_j)}_K + \norm{D(z_j) - U_j^\ast}_K
\le 2\varepsilon + L_D\norm{z_i-z_j}.
\]
\end{proof}

The bound \eqref{eq:prop1} is informative when $2\varepsilon < \norm{U_i^\ast - U_j^\ast}_K$; otherwise its right-hand side is non-positive and the inequality holds vacuously.

The premise of the proposition is that the targets themselves are
separated. The quantity measured for this premise is the minimum pairwise
distance $\norm{U_i^\ast - U_j^\ast}_K$ over the set of load cases, normalised
by the mean $K$-norm of the solutions, computed per geometry (Table~\ref{tab:pass}). Proposition~\ref{prop:one} is a conditional separation statement: it asserts that accurate decoding forces separated latent codes. It does not by itself guarantee that a given training procedure attains the required accuracy, and it does not cover latent-only objectives that use no decoder.

\paragraph{Scope limit, by counterexample.}
The statement is intrinsically restricted to a single operator: for two
different meshes, the norm on the right-hand side of \eqref{eq:prop1} is
not defined across geometries. A natural question is whether a
descriptor-free analogue holds across geometries in a suitably transported
metric. Empirically it does not. Define the direction metric
$d(A,B) = \norm{\tilde u_A - \tilde u_B}_2$, where $\tilde u$ denotes the
solution under a common load, normalised to unit Euclidean norm and
interpolated between the two meshes through coordinates normalised to a common reference domain. An explicit search over the 16 probe geometries of Section~\ref{sec:pass} finds a pair of
\emph{different} geometries whose solutions under the same load are closer
to each other, with $d_{\mathrm{cross}} = 8.39\times10^{-3}$, than the
minimum within-geometry separation between load cases on the same scale,
which is $7.56\times10^{-2}$. The descriptor-free extension across
geometries is therefore refuted. The counterexample is the intended
outcome of this check: it identifies the correct scope of the statement,
namely that separation claims across geometries must be made conditional
on geometry descriptors.

\section{The executable falsification protocol}\label{sec:pass}

Every claim above with numeric content is implemented as an executable
check and evaluated on assembled instances. Under the pre-registered
protocol, a violated inequality would have required the withdrawal of the
corresponding claim. The protocol was executed twice with identical qualitative outcomes: first on synthetic test problems in the
verification suite, and then on a probe set of 16 instances drawn from the 256-instance validation split of a pre-registered experimental run (16 July 2026) whose full
configuration was fixed in advance and identified by a cryptographic
hash. Table~\ref{tab:pass} reports the second, corpus-level measurement.

\begin{table}[ht]
\centering
\footnotesize
\begin{tabular}{@{}llll@{}}
\toprule
Claim & Checked inequality / quantity & Measured & Verdict \\
\midrule
Lemma~\ref{lem:one} & identities \eqref{eq:identities}, relative error &
$\lesssim 10^{-11}$ & holds \\
Lemma~\ref{lem:cond} & tightness $\norm{e}_2 \big/ \sqrt{2\gap/\lambda_{\min}}$,
max over probes & $0.0224$ & holds \\
 & condition-number range across instances & $\kappa \in [4.83\times10^{3},\,
1.54\times10^{5}]$ & --- \\
Cor.~\ref{cor:modes} & max deviation from \eqref{eq:modes} &
$2.11\times10^{-15}$ & exact \\
Bound \eqref{eq:cheb} & worst measured\,/\,bound over $k \in \{1,3,5,10\}$ &
$0.043$ & holds \\
Prop.~\ref{prop:one} premise & min within-geometry separation
(median; 16 geometries) & $0.464$ ($0.583$) & holds \\
Scope limit & cross-geometry counterexample vs within minimum &
$8.39\times10^{-3} < 7.56\times10^{-2}$ & refuted$^\dagger$ \\
\bottomrule
\end{tabular}
\caption{The falsification protocol on the corpus probe set (16 instances from the validation split). The entry
marked $^\dagger$ refers to the descriptor-free cross-geometry extension of
Proposition~\ref{prop:one}; its refutation is the intended outcome of the
check and selects the descriptor-conditioned form of the statement.}
\label{tab:pass}
\end{table}

\paragraph{Empirical corroboration.}
The following result provides context and is not a claim of this note. At
full scale, training on $\Pi_h$ alone, with zero reference solves, reached relative $\ell_2$ displacement parity with supervision on 1024 labelled instances, evaluated on the full 256-instance validation split (errors $0.1658$ against $0.1600$, nine random seeds across two runs), with a relative energy gap smaller by a factor of $5.1$. This is
consistent with Lemma~\ref{lem:one} and with the reading of the metrics
given by Corollary~\ref{cor:modes}.

\paragraph{Code and data availability.}
The verification suite, the stamped pre-registration documents, the experiment configurations, and the mesh-generation pipeline that reproduces the corpus are available at \url{https://github.com/HymnOfLight/FE-JEPA}. The pre-registered run of 16 July 2026 corresponds to the configuration with SHA-256 hash:\par\smallskip\noindent{\footnotesize\texttt{62b26ad868d424ef5527c8cb7d826c818aa1ba5cebbc76c7bfe665062781f0ce}}

\section{Outlook: the elastodynamic obstruction and its time-discrete repair}\label{sec:outlook}

The identities above are properties of a minimum principle. Hamilton's
principle for elastodynamics is a stationarity principle, and the
difference is quantitative. Consider the semi-discrete system, that is, the system discretised in space and continuous in time, $M\ddot u + \Kff u = F(t)$, where $M$ denotes the SPD mass matrix produced
by the same discretisation and $F(t)$ the time-dependent load. The action
functional over a time interval of length $T$ is
\[
S(u) \;=\; \int_0^{T}
\Bigl(\tfrac12\,\dot u^{\!\top} M \dot u
\;-\; \tfrac12\, u^{\!\top}\Kff\, u
\;+\; F(t)^{\!\top} u\Bigr)\,\mathrm{d}t,
\]
and the true trajectory is a stationary point of $S$. The second variation
of $S$ along the trajectory is positive semidefinite only if
$T \le \pi/\omega_{\max}$, where $\omega_{\max}$ denotes the largest
natural angular frequency of the system; that is, $\omega_{\max}^2$ is the
largest eigenvalue of the generalised eigenvalue problem
$\Kff v = \omega^2 M v$. This threshold equals one half of the shortest
natural period. On a mesh with characteristic element size $h$ one has
$\omega_{\max} = O(h^{-1})$, so the horizon on which the action is convex
along the trajectory is $O(h)$ and shrinks under mesh refinement; in the
continuum limit, the true trajectory is a saddle point of the action for
every $T > 0$ \cite{GelfandFomin1963}. Direct minimisation of the action
is therefore unsound as a training signal at the level of principle, not
of implementation.

Two exact repairs exist. The first is to discretise time before invoking a
variational principle. Allow, in addition, a damping term, so that the equation of motion reads $M\ddot u + C\dot u + \Kff u = F(t)$, where $C$ is a symmetric positive semidefinite damping matrix; the undamped case $C = 0$ is included. Every implicit step of the Newmark-$\beta$ family \cite{Newmark1959} with $\beta > 0$ then solves a linear system $\tilde A\,u = \tilde F$ with effective operator
\[
\tilde A \;=\; \Kff + \tfrac{1}{\beta\,\Delta t^{2}}\,M
+ \tfrac{\gamma}{\beta\,\Delta t}\,C,
\]
where $\Delta t$ denotes the time step and $\beta$ and $\gamma$ are the parameters of the integrator. The generalised-$\alpha$ family \cite{ChungHulbert1993} leads to an effective operator of the same form, namely a positive linear combination of $\Kff$, $M$ and $C$, and the conclusions below apply to it unchanged. The operator $\tilde A$ is SPD, and the
right-hand side $\tilde F$ is assembled from the state at the previous
step and the external load. Lemma~\ref{lem:one} therefore applies verbatim
at every step, in the norm induced by $\tilde A$, with assembly-level
evaluation. At the level of the effective operator, $\tilde A \to \Kff$ as $\Delta t \to \infty$, and the static problem is recovered provided the history-dependent effective load is treated consistently. This is the incremental-variational route, in the
tradition of variational integrators and incremental updates
\cite{MarsdenWest2001,OrtizStainier1999}. The second repair is to remain
in continuous time and to accept an additional cost. Space--time minimum
principles of the Brezis--Ekeland--Nayroles type
\cite{BrezisEkeland1976,Nayroles1976}, and elliptic regularisations of the
action whose convergence is the content of a conjecture of De Giorgi
proved by Serra and Tilli \cite{SerraTilli2012}, restore minimality at the
cost of evaluating residuals in a dual norm, which requires the
application of an operator inverse. The static energy avoids this cost
through the algebraic cancellation of Remark~\ref{rem:c2}. In this precise
sense, the static case is the degenerate case in which the dual-norm
evaluation reduces to an assembly-level computation, and the time-discrete
route inherits this property unchanged. The corresponding executable falsification checks for the dynamic identities are future work and will be pre-registered
separately.

\section{Conclusion}

For linear elastostatics on a fixed reference discretisation, the assembled
discrete potential energy is an exact label-free training objective:
minimising it is supervised regression in the stiffness norm towards the
solver's own solution, with identical gradients at every point. The
consequences developed here are the $\kappa$-mediated control of the
displacement error and the modewise contraction identity, which together
make the energy gap the primary evaluation metric; the Chebyshev guarantee
for conjugate-gradient post-processing; and a conditional latent-separation
bound on a shared operator, whose descriptor-free cross-geometry extension
is refuted by an explicit counterexample. Exactness is a property of the
objective, not of optimisation, expressivity, or generalisation, and every
numeric claim above was checked by the executable tests of
Section~\ref{sec:pass}. What remains open is stated deliberately:
nonlinear constitutive behaviour, and the elastodynamic identities of
Section~\ref{sec:outlook}, whose executable falsification checks will be
pre-registered separately.

\end{document}